\documentclass{elsarticle}

\usepackage{hyperref}
\usepackage[utf8]{inputenc}
\usepackage{amssymb}
\usepackage{amsmath}
\usepackage{amsthm}

\newtheorem{theorem}{Theorem}
\newtheorem{corollary}{Corollary}
\newtheorem{lemma}{Lemma}

\newcommand{\suf}[1]{\text{2-suf}\ensuremath{(#1)}}
\newcommand{\pref}[1]{\text{2-pre}\ensuremath{(#1)}}

\newcommand{\zso}[1]{\text{source}\ensuremath{(#1)}}

\newcommand{\zb}[1]{\text{border}\ensuremath{(#1)}}

\newcommand{\tfat}[1]{\text{fat}\ensuremath{(#1)}}

\newcommand{\rev}[1] {\text{rev}\ensuremath{(#1)}}

\newcommand{\skipi}[1] {\text{skip}\ensuremath{(#1)}}

\newcommand{\locus}[1] {\text{locus}\ensuremath{(#1)}}
\newcommand{\str}[1] {\text{str}\ensuremath{(#1)}}

\newcommand{\length}[1] {\ensuremath{|\text{str}(#1)|}}
\newcommand{\rank}[1] {\text{rank}\ensuremath{(#1)}}

\newcommand{\fat}{\ensuremath{D^{\mathrm{fat}}}}

\newcommand\floor[1]{\lfloor#1\rfloor}

\newcommand{\occ} {\ensuremath{\mathrm{occ}}}
\newcommand{\td} {\ensuremath{\mathrm{T}_D}}
\newcommand{\tdp} {\ensuremath{\mathrm{T}_{D'}}}
\newcommand{\tf} {\ensuremath{\mathrm{T}_{F}}}

\begin{document}

\begin{frontmatter}

\title{Time-Space Trade-Offs for Lempel--Ziv Compressed Indexing \tnoteref{t1}}
\tnotetext[t1]{A preliminary version of this paper appeared in the Proceedings of the 28th Annual symposium on Combinatorial Pattern Matching, 2017}

\author{Philip Bille\corref{cor1}}
\ead{phbi@dtu.dk}
\cortext[cor1]{Supported by the Danish Research Council (DFF -- 4005-00267, DFF -- 1323-00178)}
\author{Mikko Berggren Ettienne\corref{cor2}}
\ead{miet@dtu.dk}
\cortext[cor2]{Supported by the Danish Research Council (DFF -- 4005-00267)}
\author{Inge Li G{\o}rtz\corref{cor1}}
\ead{inge@dtu.dk}

\author{Hjalte Wedel Vildh{\o}j \corref{}}
\ead{hwvi@dtu.dk}

\begin{abstract}
Given a string $S$, the \emph{compressed indexing problem} is to
    preprocess $S$ into a compressed representation that supports fast
    \emph{substring queries}. The goal is to use little space relative to the
    compressed size of $S$ while supporting fast queries. We present a
    compressed index based on the Lempel--Ziv 1977 compression scheme.
    We obtain the following time-space trade-offs:
    For constant-sized alphabets
	\begin{enumerate}
            \item[(i)] $O(m + \occ \lg\lg n)$ time using $O(z\lg(n/z)\lg\lg z)$ space, or
            \item[(ii)] $O(m(1 + \frac{\lg^\epsilon z}{\lg(n/z)}) + \occ(\lg\lg n + \lg^\epsilon z))$ time using $O(z\lg(n/z))$ space, 
    \end{enumerate}
    For integer alphabets polynomially bounded by $n$
	\begin{enumerate}
        \item[(iii)] $O(m(1 + \frac{\lg^\epsilon z}{\lg(n/z)}) + \occ(\lg\lg n + \lg^\epsilon z))$ time using $O(z(\lg(n/z) + \lg\lg z))$ space, or
        \item[(iv)] $O(m + \occ(\lg\lg n + \lg^{\epsilon} z))$ time using $O(z(\lg(n/z) + \lg^{\epsilon} z))$ space,

    \end{enumerate}
    where $n$ and $m$ are the length of the input string and query string respectively,
    $z$ is the number of phrases in the LZ77 parse of the input string,
    $\occ$ is the number of occurrences of the query in the input
    and $\epsilon > 0$ is an arbitrarily small constant.
        In particular, (i) improves the leading term in the query time of the
        previous best solution from $O(m\lg m)$ to $O(m)$ at the cost of
        increasing the space by a factor $\lg \lg z$. Alternatively, (ii)
        matches the previous best space bound, but has a leading term in the query
        time of $O(m(1+\frac{\lg^{\epsilon} z}{\lg (n/z)}))$. However, for any
        polynomial compression ratio, i.e., $z = O(n^{1-\delta})$, for constant
        $\delta > 0$, this becomes $O(m)$. Our index also supports
        extraction of any substring of length $\ell$ in $O(\ell + \lg(n/z))$
        time. Technically, our results are obtained by novel extensions and
        combinations of existing data structures of independent interest,
        including a new batched variant of weak prefix search.
\end{abstract}

\end{frontmatter}

\section{Introduction}
Given a string $S$, the \emph{compressed indexing problem} is to preprocess $S$
into a compressed representation that supports fast \emph{substring queries},
that is, given a string $P$, report all
occurrences of substrings in $S$ that match $P$. Here the compressed representation
can be any compression scheme or measure ($k$th order entropy, smallest
grammar, Lempel--Ziv, etc.). The goal is to use little space relative to the
compressed size of $S$ while supporting fast queries. Compressed indexing is a
key computational primitive for querying massive data sets and the area has
received significant attention over the last decades with numerous theoretical
and practical solutions, see e.g.~\cite{Karkkainen95lempel-zivparsing, FM2000,
makinen2000, RV2000, FM2001,FM2005, GGV2003, GGV2004, FMMN2007, NM2007,
MNSV2010, CN2012, KN2014, gagie2014lz77, KarkkainenS98, BelazzouguiGGKO15} and the surveys~\cite{NM2007,
Navarro2012, Navarro2016, GagieP14}.

The Lempel--Ziv 1977 compression scheme (LZ77)~\cite{Ziv77auniversal} is a
classic compression scheme based on replacing repetitions by references in a
greedy left-to-right order. Numerous variants of LZ77 have been developed and
several widely used implementations are available (such as
\texttt{gzip}~\cite{gzip}). Recently, LZ77 has been shown to be particularly
effective at handling highly-repetitive data sets~\cite{MNSV2010, Navarro2012,
KN2014, CFMN2016, BCGPR2015} and LZ77 compression is always at least as
powerful as any grammar representation~\cite{Rytter, Charikar}.

In this paper, we consider compressed indexing based on LZ77 compression.
Relatively few results are known for this version of the problem.
Let $n$, $m$, and $z$ denote the length of the input
string, the length of the pattern string, and the number of phrases in the LZ77 parse of the string (definition follows), respectively.
K\"{a}rkk\"{a}inen and Ukkonen introduced the problem in
1996~\cite{Karkkainen95lempel-zivparsing} and gave an initial solution that required read-only access to the uncompressed text.
Interestingly, this work is among the first results in compressed
indexing~\cite{NM2007}. More recently, Gagie et al.~\cite{Gagie2012,
gagie2014lz77} revisited the problem and gave a solution using space
$O(z\lg(n/z))$ and query time $O(m \lg m + \occ \lg \lg n)$, where
$\occ$ is the number of occurrences of $P$ in $S$. Note that these bounds assume a constant-sized alphabet.

\subsection{Our Results}
We show the following main result.
\begin{theorem}
\label{theorem:main}

    We can build a compressed-index supporting substring queries in:\\
    For constant-sized alphabets
    \begin{enumerate}
        \item[(i)] $O(m + \occ \lg\lg n)$ time using $O(z\lg(n/z)\lg\lg z)$ space, or
        \item[(ii)] $O(m(1 + \frac{\lg^\epsilon z}{\lg(n/z)}) + \occ(\lg\lg n + \lg^\epsilon z))$ time using $O(z\lg(n/z))$ space,
    \end{enumerate}
    For integer alphabets polynomially bounded by $n$
	\begin{enumerate}
        \item[(iii)] $O(m(1 + \frac{\lg^\epsilon z}{\lg(n/z)}) + \occ(\lg\lg n + \lg^\epsilon z))$ time using $O(z(\lg(n/z) + \lg\lg z))$ space, or
        \item[(iv)] $O(m + \occ(\lg\lg n + \lg^{\epsilon} z))$ time using $O(z(\lg(n/z) + \lg^{\epsilon} z))$ space,
    \end{enumerate}
\noindent where $n$ and $m$ are the length of the input string and query string respectively,
$z$ is the number of phrases in the LZ77 parse of the input string,
$\occ$ is the number of occurrences of the query in the input
and $\epsilon > 0$ is an arbitrarily small constant.
\end{theorem}

Compared to the previous bounds Thm.~\ref{theorem:main} obtains new interesting
trade-offs. In particular, Thm.~\ref{theorem:main} (i) improves the leading term
in the query time of the previous best solution from $O(m\lg m)$ to $O(m)$ at
the cost of increasing the space by only a factor $\lg \lg z$. Alternatively,
Thm.~\ref{theorem:main} (ii) matches the previous best space bound, but has a
leading term in the query time of $O(m(1+\frac{\lg^{\epsilon} z}{\lg (n/z)}))$.
However, for any polynomial compression ratio, i.e., $z = O(n^{1-\delta})$, for constant
$\delta > 0$, this becomes $O(m)$.

Gagie et al.~\cite{gagie2014lz77} also showed how to extract an arbitrary
substring of $S$ of length $\ell$ in time $O(\ell + \lg n)$. We show how to
support the same extraction operation and slightly improve the time to $O(\ell + \lg
(n/z))$.

Technically, our results are obtained by new variants and extensions of
existing data structures in novel combinations. In particular, we consider a
batched variant of the \emph{weak prefix search problem} and give the first
non-trivial solution to it. We also generalize the well-known bidirectional
compact trie search technique~\cite{Lewenstein13} to reduce the number of
queries at the cost of increasing space. Finally, we show how to combine this
efficiently with range reporting and fast random-access in a balanced grammar
leading to the result.

We note that none of our data structures assume constant-sized alphabet
and therefore, Thm.~\ref{theorem:main} is an instance of a full time-space trade-off
for general alphabets. We discuss the details in Sec.~\ref{combining}.

\section{Preliminaries}

\label{section:pre}
We assume a standard unit-cost RAM model with
word size $w = \Theta(\lg n)$ and that the input is from an integer alphabet $\Sigma = \{1, 2, \ldots, n^{O(1)}\}$
and measure space complexity in words unless otherwise specified.

A string $S$ of length $n = |S|$ is a sequence $S[1]\cdots S[n]$ of $n$ characters drawn from~$\Sigma$.
The string $S[i]\cdots S[j]$ denoted $S[i,j]$ is called a \textit{substring} of $S$.
$\epsilon$ is the empty string and $S[i,i] = S[i]$ while $S[i,j] = \epsilon$ when $i > j$.
The substrings $S[1,i]$ and $S[j,n]$ are the $i^{th}$
\textit{prefix} and the $j^{th}$ \textit{suffix} of $S$ respectively. The reverse
of the string $S$ is denoted $\rev{S} = S[n]S[n-1]\cdots S[1]$.
We use the results from Fredman et al.~\cite{Fredman} when referring to perfect
hashing allowing us to build a dictionary on $O(k)$ integer keys
in $O(k)$ expected time supporting constant time lookups.

\subsection{Compact Tries}

A \textit{trie} for a set $D$ of $k$ strings
is a rooted tree where the vertices correspond to the prefixes
of the strings in $D$.
$\str{v}$ denotes the prefix corresponding to the vertex $v$.
$\str{v} = \epsilon$ if $v$ is the root while $v$ is the parent of $u$
if $\str{v}$ is equal to $\str{u}$ without the last character.
We may use $v$ in place of $\str{v}$ when it is clear from the
context that we talk about the $\str{v}$.
This character is then the $\textit{label}$ of the edge from $u$ to $v$.
The \textit{depth} of vertex $v$ is the number of edges on the path from $v$ to the root.

We assume each string in $D$ is terminated by a special character $\$ \notin \Sigma$ such that each string in $D$
corresponds to a leaf.
The children of each vertex are sorted from left to right in increasing lexicographical order,
and therefore the left to right order of the leaves corresponds
to the lexicographical order of the strings in $D$.
Let \rank{s} denote the rank of the string $s \in D$ in this order.

A \textit{compact trie} for $D$ denoted \td\ is obtained from the trie by
removing all vertices $v$ with exactly one child excluding the root
and replacing the two edges incident to $v$ with a single edge from its
parent to its child. This edge is then labeled with the concatenation of
the edge labels it replaces, thus the edges of a compact trie may be labeled by strings.
The \textit{skip interval} of a vertex $v \in \td$ with parent $u$
is $(\length{u},\length{v}]$ denoted $\skipi{v}$ and $\skipi{v} = \emptyset$ if $v$ is the root.
The \textit{locus} of a string $s$ in $\td$, denoted \locus{s},
is the minimum depth vertex $v$ such that $s$ is a prefix of \str{v}.
If there is no such vertex, then $\locus{s} = \bot$.

In order to reduce the space used by \td\ we
only store the first character of every edge and in every vertex
$v$ we store $\length{v}$ (This variation is also known as a PATRICIA tree~\cite{Morrison}).
We navigate $\td$ by storing a dictionary in every internal vertex
mapping the first character of the label of an edge
to the respective child. The size of \td\ is $O(k)$.

\subsection{Karp--Rabin Fingerprints}

A \emph{Karp--Rabin fingerprinting function}~\cite{Karp:1987} is a randomized hash function for strings.
The fingerprint for a string $S$ of length $n$ is defined as:
\[
    \phi(S) = \sum^{n}_{i=1} S[i] \cdot r^{i-1} \text{ mod } p
\]
where $p$ is a prime and $r$ is a random integer in $\mathbb{Z}_{p}$ (the field of integers modulo $p$).
Storing the values $n$, $r^n \text{ mod } p$ and $r^{-n} \text{ mod } p$ along with
a fingerprint allows for efficient composition and subtraction of fingerprints:
\begin{lemma}
\label{karp1}
    Let $x,y,z$ be strings such that $x = yz$. Given two of the three fingerprints
    $\phi(x), \phi(y)$ and $\phi(z)$, the third can be computed in constant time.
\end{lemma}

It follows that we can compute and store the fingerprints
of each of the prefixes of a string $S$ of length $n$ in $O(n)$ time and space such that we afterwards
can compute the fingerprint of any substring $S[i,j]$ in constant time.
We say that the fingerprints of the strings $x$ and $y$ \textit{collide} when $\phi(x) = \phi(y)$ and $x \neq y$.
A fingerprinting function $\phi$ is \textit{collision-free} for a set of strings if there are no fingerprint collisions
between any of the strings.

\begin{lemma}\label{karp2}
    Let $x$ and $y$ be different strings of length at most $n$ and let $p = \Theta(n^{2+\alpha})$
    for some $\alpha > 0$.
    The probability that $\phi(x) = \phi(y)$ is $O(1/n^{1+\alpha})$.
\end{lemma}

\noindent See e.g.~\cite{Porat,Breslauer} for proofs of Lemma~\ref{karp1} and~\ref{karp2}.

\subsection{Range Reporting}

Let $X \subseteq {\{0,\ldots, u\}}^d$ be a set of points in a d-dimensional grid. The \textit{orthogonal range reporting problem}
in $d$-dimensions is to compactly represent $X$ while supporting \textit{range reporting queries}, that is, given
a rectangle $R = [a_1, b_2] \times \cdots \times [a_d, b_d]$ report all points in the set $R \cap X$.
We use the following results for 2-dimensional range reporting:

\begin{lemma}[Chan et al.~\cite{Larsen11}]
\label{lemma:larsen}
    For any set of $n$ points in $[0,u] \times [0,u]$ and $2 \leq B \leq \lg^\epsilon n, 0 < \epsilon < 1$ we can solve 2-d orthogonal range reporting with
    $O(n\lg n)$ expected preprocessing time, $O(n \lg_B\lg n)$ space and $(1+k) \cdot O(B\lg\lg u)$ query time
    where $k$ is the number of occurrences inside the rectangle.
\end{lemma}

\subsection{LZ77}

The Ziv--Lempel algorithm from 1977~\cite{Ziv77auniversal} provides a simple and
natural way to compress strings.

The \textit{LZ77 parse} of a string $S$ of length $n$ is a sequence $Z$ of
$z$ subsequent substrings of $S$ called \textit{phrases} such that $S = Z[1]Z[2]\cdots Z[z]$.
$Z$ is constructed in a left to right pass of $S$:
Assume that we have found the sequence $Z[1, i]$ producing the string $S[1, j - 1]$
and let $S[j, j'-1]$ be the longest prefix of $S[j, n-1]$
that is also a substring of $S[1, j' - 2]$. Then $Z[i+1] = S[j, j']$.
The occurrence of $S[j, j' - 1]$ in $S[1, j'-2]$ is called
the \textit{source} of the phrase $Z[i]$.
Thus a phrase is composed by the contents of its possibly empty
source and a trailing character which we call the \textit{phrase border} and is typically represented as a
triple $Z[i] = (start, len, c)$ where \textit{start} is the
starting position of the source, \textit{len} is
the length of the source and $c \in \Sigma$ is the border.
For a phrase $Z[i] = S[j, j']$ we denote the position of its
border by $\zb{Z[i]} = j'$ and its source by $\zso{Z[i]} = S[j, j' - 1]$.
For example, the string $abcabcabc\cdots abc$ of length $n$ has the LZ77 parse
$|a|b|c|abcabcabc\cdots abc|$ of length 4 which is represented as $Z = (0,0,a)(0,0,b)(0,0,c)(1,n-4,c)$.

\section{Prefix Search}
\label{section:fat}

The \textit{prefix search} problem is
to preprocess a set of strings such that later, we can find
all the strings in the set that are prefixed by some query string.
Belazzougui et al.~\cite{BelazzouguiBPV10} consider the \textit{weak prefix search} problem,
which is a relaxation of the prefix search problem. In this variant,
we report only the ranks (in lexicographic order) of the strings that are prefixed by the
query pattern and we are only required to answer correctly when at least
one of the strings is prefixed by the pattern.
Thus we may answer arbitrarily when no strings are prefixed by the query pattern.

\begin{lemma}[Belazzougui et al.~{\cite[appendix H.3]{BelazzouguiBPV10}}]
\label{fat1}
    Given a set $D$ of $k$ strings with average length $l$, from an alphabet of size $\sigma$,
    we can build a data structure using $O(k(\lg l + \lg\lg \sigma))$ bits of space
    supporting weak prefix search for a pattern $P$ of length $m$ in $O(m\lg \sigma/w + \lg m)$ time
    where $w$ is the word size.
\end{lemma}
The term $m\lg \sigma/w$ stems from preprocessing $P$ with an incremental hash function
such that the hash of any substring $P[i,j]$ can be obtained in constant time afterwards.
Therefore we can do weak prefix search for $h$ substrings of $P$ in $O(m\lg \sigma/w + h\lg m)$ time.
We now describe a data structure that builds on the ideas from Lemma~\ref{fat1} but obtains the following:

\begin{lemma}
\label{fast-fat}
    Given a set $D$ of $k$ strings, we can build a data structure
    taking $O(k)$ space supporting weak prefix search for $h$ substrings of a pattern $P$ of length $m$
    in time $O(m + h(m/x + \lg x))$ where $x$ is a positive integer.
\end{lemma}
If we know $h$ when building our data structure,
we set $x$ to $h$ and obtain a query time of $O(m + h\lg h)$ with Lemma~\ref{fast-fat}.

Before describing our data structure we need the following definition:
The \textit{2-fattest} number in a nonempty interval of strictly positive integers is the unique number
in the interval whose binary representation has the highest number of trailing zeroes.

\subsection{Data Structure}
Let \td\ be the compact trie representing the set $D$ of $k$ strings and let $x$ be a positive integer.
Denote by \tfat{v} the 2-fattest number in the skip interval of a vertex $v \in \td$.
The \textit{fat prefix} of $v$ is the length \tfat{v}\ prefix of $\str{v}$.
Denote by \fat\ the set of fat prefixes induced by the vertices of $\td$.
The $x$-prefix of $v$ is the shortest prefix of $\str{v}$
whose length is a multiple of $x$ and is in the interval \skipi{v}.
If $v$'s skip interval does not span a multiple of $x$, then $v$ has no $x$-prefix.
Let $D^x$ be the set of $x$-prefixes induced by the vertices of $\td$.
The data structure is the compact trie $\td$ augmented with:
\begin{itemize}

    \item A fingerprinting function $\phi$.

    \item A dictionary $\mathcal{G}$ mapping the fingerprints of the strings in \fat\ to their associated vertex.

    \item A dictionary $\mathcal{H}$ mapping the fingerprints of the strings in $D^x$ to their associated vertex.

    \item For every vertex $v \in \td$ we store the rank in $D$ of the string represented
        by the leftmost and rightmost leaf in the subtree of $v$, denoted $l_v$ and $r_v$ respectively.

\end{itemize}
The data structure is similar to the one by Belazzougui~et al.~\cite{BelazzouguiBPV10} except for the dictionary $\mathcal{H}$, which we use in the first step of our search.

There are $O(k)$ strings in each of \fat\ and $D^x$
thus the total space of the data structure is $O(k)$.

Let $i$ be the start of the skip interval of some vertex $v \in \td$
and define the \textit{pseudo-fat} numbers of $v$ to
be the set of 2-fattest numbers in the intervals $[i,p]$
where $i \leq p < \tfat{v}$.
We use Lemma~\ref{karp2} to find a fingerprinting function $\phi$ that
is collision-free for the strings in \fat\
and all length $l$ prefixes of the strings in $D$ where $l$ is
either a pseudo-fat number in the skip interval of some vertex $v \in \td$
or a multiple of $x$.

Observe that the range of strings in $D$ that are prefixed by some pattern $P$ of length $m$ is
exactly $[l_v, r_v]$ where $v = \locus{P}$.
Answering a weak prefix search query for $P$ is comprised
by two independent steps. First step is to find a vertex $v \in \td$ such that
\str{v} is a prefix of $P$ and $m - \length{v} \leq x$. We say that $v$ is in \emph{$x$-range of $P$}.
Next step is to apply a slightly modified version of the search technique from Belazzougui
et al.~\cite{BelazzouguiBPV10}
to find the \textit{exit vertex} for $P$, that is, the deepest vertex $v_e \in \td$ such
that $\str{v_e}$ is a prefix of $P$.
Having found the exit vertex we can find the locus in constant time
as it is either the exit vertex itself or one of its children.

\subsection{Finding an \texorpdfstring{$x$}{x}-range Vertex}

We now describe how to find a vertex in $x$-range of $P$. 
If $m < x$ we simply report that the root of $\td$ is in $x$-range of $P$.
Otherwise, let $v$ be the root of $\td$ and for $i = 1, 2, \ldots \floor{m/x}$
we check if $ix > \length{v}$ and $\phi(P[1,ix])$ is in $\mathcal{H}$
in which case we update $v$ to be the corresponding vertex.
Finally, if $\length{v} \geq m$ we report that $v$ is $\locus{P}$
and otherwise we report that $v$ is in $x$-range of $P$.
In the former case, we report $[l_v, r_v]$
as the range of strings in $D$ prefixed by $P$. In the latter case we pass on $v$ to the next step
of the algorithm.

We now show that the algorithm is correct when $P$ prefixes a string in $D$. It is easy to verify that the $x$-prefix of $v$ prefixes $P$ at all time during the execution of the algorithm.
Assume that $\length{v} \geq m$ by the end of the algorithm. We will show that in that case $v = \locus{P}$, i.e., that $v$ is the highest vertex prefixed by $P$.  Since $P$ prefixes a string in $D$, the $x$-prefix of $v$ prefixes $P$, and $\length{v} \geq m$, then $P$ prefixes $v$. Since  the $x$-prefix of $v$ prefixes $P$, $P$ does not prefix the parent of $v$ and thus  $v$ is the highest vertex prefixed by $P$.

Assume now that $\length{v}<m$. We will show that  $v$ is in $x$-range of $P$. Since $P$ prefixes a string in $D$ and the $x$-prefix of $v$ prefixes $P$, then $\str{v}$ prefixes $P$. Let $P[1,ix]$ be the $x$-prefix of $v$. Since $v$ is returned,  either $\phi(P[1,jx])\not\in \mathcal{H}$  or $jx \leq \length{v}$ for all $i<j \leq \floor{m/x}$. If $\phi(P[1,jx])\not\in \mathcal{H}$ then $P[1,jx]$ is not a $x$-prefix of any vertex in $\td$.  Since $P$ prefixes a string in $D$ this implies that $jx$ is in the skip interval of $v$, i.e., $jx \leq \length{v}$.
This means that  $jx \leq \length{v}$ for all $i<j \leq \floor{m/x}$.  Therefore  $\floor{m/x}x\leq \length{v} < m$ and it  follows that $m - \length{v} < x$. We already proved that  $\str{v}$ prefixes $P$  and therefore $v$ is in $x$-range of~$P$.

In case $P$ does not prefix any string in $D$
we either report that $v = \locus{P}$ even though $\locus{P} = \bot$
or report that $v$ is in $x$-range of $P$ because $m - \length{v} \leq x$  even though $\str{v}$
is not a prefix of $P$ due to fingerprint collisions.
This may lead to a false positive. 
However, false positives are allowed 
in the weak prefix search problem.

Given that we can compute the fingerprint of substrings of $P$ in constant time
the algorithm uses $O(m/x)$ time. 

\subsection{From \texorpdfstring{$x$}{x}-range to Exit Vertex}

We now consider how to find the exit vertex of $P$ hereafter denoted $v_e$.
The algorithm is similar to the one presented in Belazzougui et al.~\cite{BelazzouguiBPV10} except that
we support starting the search from not only the root, but from any ancestor of $v_e$.

Let $v$ be any ancestor of $v_e$, let $y$ be the smallest power of two greater than $m - \length{v}$
and let $z$ be the largest multiple of $y$ no greater than $\length{v}$.
The search progresses by iteratively halving the search interval
while using $\mathcal{G}$ to maintain a candidate for the exit vertex and
to decide in which of the two halves to continue the search.

Let $v_c$ be the candidate for the exit vertex and let $l$ and $r$ be the left and right boundary for our search interval.
Initially  $v_c = v$,  $l = z$ and $r = z + 2y$.
When $r-l = 1$, the search terminates and reports $v_c$. In each iteration, we consider
the mid $b = (l+r)/2$ of the interval $[l, r]$ and update the interval to either $[b, r]$ or $[l, b]$.
There are three cases: 
\begin{enumerate}
    \item $b$ is out of bounds
        \begin{enumerate}
            \item If $b > m$ set $r$ to $b$.
            \item If $b \leq \length{v_c}$ set $l$ to $b$.
        \end{enumerate}

    \item $P[1,b] \in \fat$, let $u$ be the corresponding vertex, i.e.\ $\mathcal{G}(\phi(P[1,b])) = u$.
        \begin{enumerate}
            \item If  $\length{u} < m$, set $v_c$ to $u$ and $l$ to $b$.
            \item If $\length{u} \geq m$, report $u = \locus{P}$ and terminate.
        \end{enumerate}

    \item $P[1,b] \notin \fat$ and thus $\phi(P[1,b])$ is not in $\mathcal{G}$, set $r$ to $b$.

\end{enumerate}

Observe that we are guaranteed that all fingerprint comparisons are collision-free in case
$P$ prefixes a string in $D$.
This is because the length of the prefix fingerprints we consider are all either 2-fattest or pseudo-fat
in the skip interval of $\locus{P}$ or one of its ancestors and we use a fingerprinting function
that is collision-free for these strings.

\subsubsection{Correctness}
We now show that the invariant $l \leq \length{v_c} \leq \length{v_e} < r$ is satisfied and that
$\str{v_c}$ is a prefix of $P$ before and after each iteration.
After $O(\lg x)$ iterations $r-l = 1$ and thus $l = \length{v_e} = \length{v_c}$ and therefore $v_c = v_e$.
Initially $v_c$ is an ancestor of $v_e$ and thus $\str{v_c}$ is a prefix of $P$,
$l = z \leq \length{v_c}$ and $r = z + 2y > m > \length{v_e}$ so the invariant is true.
Now assume that the invariant is true at the beginning of some iteration and consider the possible cases:
\begin{enumerate}
    \item $b$ is out of bounds
        \begin{enumerate}
            \item $b > m$ then because $\length{v_e} \leq m$, setting $r$ to $b$
        preserves the invariant.
            \item $b \leq \length{v_c}$ then setting $l$ to $b$ preserves the invariant.
        \end{enumerate}

    \item $P[1, b] \in \fat$, let $u = \mathcal{G}(\phi(P[1,b]))$.
        \begin{enumerate}
            \item $\length{u} < m$
                then $\str{u}$ is a prefix of $P$ and thus $b = \tfat{u} \leq \length{u} \leq \length{v_e}$
                so setting $l$ to $b$ and $v_c$ to $u$ preserves the invariant.

            \item $\length{u} \geq m$ yet $u = \mathcal{G}(\phi(P[1,b]))$.
                Then $u$ is the locus of $P$.
        \end{enumerate}

    \item $P[1, b] \notin \fat$, and thus $\phi(P[1,b])$ is not in $\mathcal{G}$.
        As we are not in any of the out of bounds cases we have $\length{v_c} < b < m$.
        Thus, either $b > \length{v_e}$ and setting $r$ to $b$ preserves the invariant.
        Otherwise $b \leq \length{v_e}$ and thus $b$ must be in the skip interval of some
        vertex $u$ on the path from $v_c$ to $v_e$ excluding $v_c$.
        But $\skipi{u}$ is entirely included in $(l, r)$ and because
        $b$ is 2-fattest in $(l,r)$\footnote{If $b-a = 2^i$, $i > 0$ and $a$ is a multiple of $2^{i-1}$
            then the mid of the interval $(a + b)/2$ is 2-fattest in $(a,b)$.
        } it is also 2-fattest in $\skipi{u}$.
        It follows that $\tfat{u} = b$ which contradicts $P[1, b] \notin \fat$
        and thus the invariant is preserved.
\end{enumerate}

Thus if $P$ prefixes a string in $D$ we find either the exit vertex $v_e$ or the locus of $P$.
In the former case the locus of $P$ is the child of $v_e$ identified by the character $P[\length{v_e} + 1]$.
Having found the vertex $u = \locus{P}$ we report $[l_{u}, r_{u}]$ as the range of strings in $D$ prefixed by $P$.
In case $P$ does not prefix any strings in $D$, the fact that the fingerprint
of a prefix of $P$ match the fingerprint of some fat prefix in $D^x$
does not guarantee equality of the strings.
There are two possible consequences of this. Either the search successfully
finds what it believes to be the locus of $P$ even though $\locus{P} = \bot$
in which case we report a false positive. Otherwise,
there is no child identified by $P[\length{v_e} + 1]$ in which case we can correctly report
that no strings in $D$ are prefixed by $S$, a true negative.
Recall that false positives are allowed as we are considering the weak prefix search problem.

\subsubsection{Complexity}
The size of the interval $[l, r]$ is halved in each iteration, thus
we do at most $O(\lg(m - \length{v}))$ iterations, where $v$ is the
vertex from which we start the search.
If we use the technique from the previous section to find a starting vertex
in $x$-range of $P$, we do $O(\lg x)$ iterations.
Each iteration takes constant time.
Note that if $P$ does not prefix a string in $D$ we may have fingerprint collisions
and we may be given a starting vertex $v$
such that $\str{v}$ does not prefix $P$.
This can lead to a false positive, but we still have $m - \length{v} \leq x$ and therefore the time complexity remains $O(\lg x)$.

\subsection{Multiple Substrings}
In order to answer weak prefix search queries for $h$ substrings
of a pattern $P$ of length $m$, we first preprocess $P$ in $O(m)$ time such that we can compute the
fingerprint of any substring of $P$ in constant time using Lemma~\ref{karp1}.
We can then answer a weak prefix search query for any substring of $P$
in total time $O(m/x + \lg x)$ using the techniques described in the previous sections.
The total time is therefore $O(m + h(m/x + \lg x))$. 

\section{Distinguishing Occurrences}

\label{section:pmlz77}

The following sections describe our compressed-index consisting of three independent data structures.
One that finds long primary occurrences, one that finds short primary occurrences
and one that finds secondary occurrences.

Let $Z$ be the LZ77 parse of length $z$ representing the string $S$ of length $n$.
If $S[i,j]$ is a phrase of $Z$ then any substring of $S[i,j-1]$ is a \textit{secondary substring} of $S$.
These are the substrings of $S$ that do not contain any phrase borders.
On the other hand, a substring $S[i', j']$ is a \textit{primary substring} of $S$
when there is some phrase $S[i,j]$ where $i \leq i' \leq j \leq j'$,
these are the substrings that contain one or more phrase borders. Any substring of $S$ is either primary or secondary.
A primary substring that matches a query pattern $P$ is a \textit{primary occurrence} of $P$ while a secondary substring that
matches $P$ is a \textit{secondary occurrence}~\cite{Karkkainen95lempel-zivparsing}.

\section{Long Primary Occurrences}

For simplicity, we assume that the data structure given in Lemma~\ref{fast-fat}
not only solves the weak prefix problem, but also answers correctly when
the query pattern does not prefix any of the indexed strings.
Later in Section~\ref{verification} we will see how to lift this assumption.
The following data structure and search algorithm
is a variation of the classical bidirectional search technique for finding primary occurrences~\cite{Karkkainen95lempel-zivparsing}.

\subsection{Data Structure}
\label{main-ds}

Let $\tau$ be a fixed positive integer parameter (its value will be determined later).
For every phrase $S[i,j]$, we consider the strings $S[i,j+k], 0 \leq k < \tau$
\textit{relevant substrings} of $S$ unless there is some longer relevant
substring ending at position $j+k$.
If $S[i',j']$ is a relevant substring then the string $S[j'+1, n]$ is the \textit{associated suffix}.
There are at most $z\tau$ relevant substrings of $S$ and equally many associated suffixes.
The primary index is comprised by the following:

\begin{itemize}
    \item A prefix search data structure \td\ 
    on the set of reversed relevant substrings.
    \item A prefix search data structure \tdp\ 
    on the set of associated suffixes.
    \item An orthogonal range reporting data structure $R$  on the $z\tau \times z\tau$ grid.
        Consider a relevant substring $S[i,j]$.
        Let $x$ denote the rank of $\rev{S[i,j]}$ in the lexicographical order of the reversed relevant substrings,
        let $y$ denote the rank of its associated suffix $S[j+1,n]$ in the lexicographical order
        of the associated suffixes.
        Then $(x,y)$ is a point in $R$ and along with it we store the pair $(j,b)$, where $b$ is the position of the rightmost phrase border contained in $S[i,j]$. 
\end{itemize}

Note that every point $(x,y)$ in $R$ is induced by some relevant substring $S[i, j]$ and its associated suffix $S[j+1, n]$.
If some prefix $P[1, k]$ is a suffix of $S[i,j]$ and the suffix $P[k+1, m]$ is a prefix of $S[j+1,n]$
then $S[j - k + 1, j - k + m]$ is an occurrence of $P$ and we can compute its exact location from $k$ and $j$.

\subsection{Searching}
\label{primary}

The data structure can be used to find the primary occurrences of a pattern $P$ of length $m$ when $m > \tau$.
Consider the $O(m/\tau)$ prefix-suffix pairs $(P[1,i\tau], P[i\tau+1, m])$ for $i = 1, \ldots ,\floor{m/\tau}$
and the pair $(P[1,m], \epsilon)$ in case $m$ is not a multiple of $\tau$.
For each such pair, we do a prefix search for \rev{P[1,i\tau]} and $P[i\tau+1, m]$
in $\td$ and $\tdp$, respectively. If either of these two searches report no matches, we move on to the next pair.
Otherwise, let $[l,r]$, $[l',r']$ be the ranges reported from the search
in $\td$ and $\tdp$ respectively. Now we do a range reporting query on $R$
for the rectangle $[l,r] \times [l',r']$. For each point reported, let $(j, b)$ be the pair stored with the point.
We report $j-i\tau + 1$ as the starting position of a primary occurrence of $P$ in $S$.

Finally,  in case $m$ is not a multiple of $\tau$, we need to also check the pair $(P[1,m], \epsilon)$. We search for  \rev{P[1,m]} in
in $\td$ and  $\epsilon$ in $\tdp$. If the search for \rev{P[1,m]}  reports  no match we stop. Otherwise, we do a range reporting query as before. For each point reported, let $(j, b)$ be the pair stored with the point. To check that the occurrence has not been reported before we do as follows.
Let $k$ be the smallest positive integer such that $j-m+k\tau>b$. Only if $k\tau> m$ we report $j-m+1$ as the  starting position of a primary occurrence. 

\subsubsection{Correctness}
We claim that the reported occurrences are exactly the primary occurrences of $P$.
We first prove that all primary occurrences are reported correctly. Let $P = S[i',j']$ be a primary occurrence.
As it is a primary occurrence, there must be some phrase $S[i^*, j^*]$ such that $i^* \leq i' \leq j^* \leq j'$.
Let $k$ be the smallest positive integer such that $i' + k\tau - 1 \geq j^*$.
There are two cases: $k\tau \leq m$ and $k\tau >m$. If $k\tau \leq m$ then $P[1, k\tau]$ is  a suffix of the relevant substring ending at $i'+k\tau-1$. Such a relevant substring exists since $i'+k\tau-1 < j^*+\tau$.  Thus its reverse $\rev{P[1, k\tau]}$
prefixes a string $s$ in $D$, while $P[k\tau + 1, m]$ is a prefix of the associated suffix $S[i'+k\tau, n] \in D'$.
Therefore, the  respective ranks of $s$ and $S[i'+k\tau, n] $ in $D$ and $D'$ are plotted as a point in $R$ which stores the pair $(i' + k\tau - 1,b)$.
We will find this point when considering the prefix-suffix pair $(P[1, k\tau]$, $P[k\tau + 1, m])$,
 and correctly report $(i'+k\tau - 1) - k\tau + 1 = i'$ as the starting position of a primary occurrence.
If $k\tau >m$ then $P[1,m]$ is a suffix of the relevant substring ending in $i'+m-1$. Such a relevant substring exists since $i' + m - 1 < i'+k\tau-1 < j^*+\tau$.  Thus its reverse prefixes a string in $D$ and trivially $\epsilon$ is a prefix of the associated suffix. It follows as before that the ranks are plotted as a point in $R$ storing the pair $(i' + m - 1,b)$ and that we find this point when considering the pair $(P[1, m], \epsilon)$. When considering $(P[1, m], \epsilon)$ we report  $(i'+m - 1) - m + 1 = i'$ as the starting position of a primary occurrence if $k\tau > m$, and thus $i'$ is correctly reported.

We now prove that all reported occurrences are in fact primary occurrences. Assume that we report $j-i\tau + 1$ for some $i$ and $j$ as the starting position of a primary occurrence in the first part of the procedure.
Then there exist strings $\mathrm{rev}(S[i', j])$ and $S[j+1, n]$ in $D$ and $D'$ respectively
such that $S[i', j]$ is suffixed by $P[1,i\tau]$ and $S[j+1, n]$ is prefixed by $P[i\tau + 1, m]$.
Therefore $j-i\tau + 1$ is the starting position of an occurrence of $P$.
The string $S[i',j]$ is a relevant suffix and therefore there exists a border $b$ in the interval $[j-\tau +1,j]$. Since $i\geq 1$  the occurrence contains the border $b$ and it is therefore a primary occurrence. If we report  $j-m + 1$ for some $j$ as the starting position of a primary occurrence in the second part of the procedure, then $\mathrm{rev}(P[1, m])$ is a prefix of a string $\rev{S[i',j]}$ in $D$. It follows immediately that $j-m+1$ is the starting point of an occurrence. Since $m>\tau$ we have $j-m+1 < j-\tau +1$, and by the definition of relevant substring there is a border in the interval $[j-\tau+1,j]$. Therefore the occurrence contains the border and is primary.

\subsubsection{Complexity}
We now consider the time complexity of the algorithm described.
First we will argue that any primary occurrence is reported at most once and that the search finds at most two points in $R$ identifying it.
Let $S[i',j']$ be a primary occurrence reported when we considered the prefix-suffix pair $(P[1, k\tau]$, $P[k\tau + 1, m])$ as in the proof of correctness.
Recall that there is some phrase $S[i^*, j^*]$ such that $i^* \leq i' \leq j^* \leq j'$ and again let $k$ be the smallest positive integer such that $i' + k\tau - 1 \geq j^*$.
None of the pairs $(P[1,h\tau], P[h\tau + 1, m])$, where $1 \le h < k$ will identify this occurrence as 
the reverse of $P[1,h\tau]$ does not prefix the reverse of any relevant substring when $i^* \le i' \le i'+h\tau -1<j^*$ which is true when $h < k$.
None of the pairs $(P[1,h\tau], P[h\tau + 1, m])$, where $h > k$, will identify this occurrence. This is the case since $i' + h\tau - 1 > j^*+\tau - 1$ whenever $h > k$, and from the definition of relevant substrings it follows that if $S[i^*,j^*]$ is a phrase, $S[a, b]$ is a relevant substring and $a < i^*$, then $b < i^* + \tau - 1$. Thus there are no relevant substrings that end after $j^* + \tau - 1$ and start before $i' < j^*$.
Therefore, only one of the pairs $(P[1, h\tau], P[h\tau + 1, m])$ for $h = 1, \ldots \floor{m/\tau}$ identifies the occurrence.
If $(k+1)\tau > m$ then we might also find the occurrence when considering the pair  $(P[1, m], \epsilon)$, but we do not report $i'$ as $k\tau\leq m$.

After preprocessing $P$ in $O(m)$ time,
we can do the $O(m/\tau)$ prefix searches in total time $O(m + m/\tau(m/x + \lg x))$ where $x$ is a positive integer by Lemma~\ref{fast-fat}.
Using the range reporting data structure by Chan~et~al.~\cite{Larsen11} each range reporting query takes
$(1+k)\cdot O(B\lg \lg (z\tau))$ time where $2 \leq B \leq \lg^\epsilon (z\tau)$ and $k$ is the number of points reported.
As each such point in one range reporting query corresponds to the identification of a unique primary  occurrence of $P$,
which happens at most twice for every occurrence we charge $O(kB\lg\lg(z\tau))$ to reporting the occurrences.
The total time to find all primary occurrences is thus $O(m + \frac{m}{\tau}(\frac{m}{x} + \lg x + B\lg\lg(z\tau)) + \occ\  B\lg\lg(z\tau))$
where $\occ$ is the number of primary and secondary occurrences of $P$.

\subsection{Prefix Search Verification}
\label{verification}

The prefix data structure from Lemma~\ref{fast-fat}
gives no guarantees of correct answers when the query pattern does not prefix any of the indexed strings.
If the prefix search gives false-positives, we may end up reporting occurrences of $P$
that are not actually there. We show how to solve this problem after introducing a series of tools that we will need.

\subsubsection{Straight Line Programs}

\label{slp}

A \textit{straight line program} (SLP) for a string $S$ is a
context-free grammar generating the single string $S$.

\begin{lemma}[Rytter~\cite{Rytter}, Charikar~et~al.~\cite{Charikar}]
\label{rytter}
        Given an LZ77 parse $Z$ of length $z$ producing
        a string $S$ of length $n$ we can construct
        a SLP for $S$ of size $O(z \lg (n/z))$ in time
        $O(z\lg (n/z))$.
\end{lemma}

The construction from Rytter~\cite{Rytter} produces a balanced grammar for every consecutive
substring of length $n/z$ of $S$ after a preprocessing step transforms $Z$ such that no compression element
is longer than $n/z$.
These grammars are then connected to form a single balanced grammar of height $O(\lg n)$ which
immediately yields extraction of any substring $S[i,j]$ in time $O(\lg(n) + j-i)$.
We give a simple solution to reduce this to $O(\lg(n/z) + j - i)$,
that also supports computation of the fingerprint of a substring in
$O(\lg(n/z))$ time.

\begin{lemma}
\label{small-slp}
    Given an LZ77 parse $Z$ of length $z$ producing a string
    $S$ of length $n$ we can build a data structure that for any substring $S[i,j]$
    can extract $S[i, j]$ in $O(\lg(n/z) + j - i)$
    time and compute the fingerprint $\phi(S[i,j])$ in $O(\lg(n/z))$ time.
    The data structure uses $O(z\lg(n/z))$ space and $O(n)$ construction time.
\end{lemma}

\begin{proof}
Assume for simplicity that $n$ is a multiple of $z$.
We construct the SLP producing $S$ from $Z$.
Along with every non-terminal of the SLP we store the size and fingerprint of its expansion.
Let $s_1, s_2, \ldots s_z$ be consecutive length $n/z$ substrings of $S$.
We store the balanced grammar producing $s_i$ along with the fingerprint $\phi(S[1, (i-1)n/z])$ at index $i$ in a table~$A$.

Now we can extract $s_i$ in $O(n/z)$ time and any substring $s_i[j,k]$ in time $O(\lg(n/z) + k-j)$.
Also, we can compute the fingerprint $\phi(s_i[j, k])$ in $O(\lg(n/z))$ time.
We can easily do a constant time mapping from a position in $S$
to the grammar in $A$ producing the substring covering that position
and the corresponding position inside the substring.
But then any fingerprint $\phi(S[1, j])$ can be computed in time $O(\lg(n/z))$.
Now consider a substring $S[i,j]$ that starts in $s_k$ and ends in $s_l, k < l$.
We extract $S[i, j]$ in $O(\lg(n/z) + j-i)$ time by extracting
the appropriate suffix of $s_k$, all of $s_m$ for $k < m < l$ and the appropriate prefix of $s_l$.
Each of the fingerprints stored by the data structure can be computed in $O(1)$
time after preprocessing $S$ in $O(n)$ time. Thus table $A$ is filled in $O(z)$ time
and by Lemma~\ref{rytter} the SLPs stored in $A$ use a total of $O(z\lg(n/z))$ space and construction time.
\end{proof}

\subsubsection{Verification of Fingerprints} We need the following lemma for the verification.
\begin{lemma}[Bille et al.~\cite{Bille2012}]
\label{2power}
    Given a string $S$ of length $n$, we can find a fingerprinting function $\phi$
    in $O(n\lg n)$ expected time such that 
    \[\phi(S[i, i+2^l]) = \phi(S[j, j+2^l]) 
    \text{ iff } S[i,i+2^l] = S[j, j+2^l] \text{ for all } (i,j,l). \]

\end{lemma}

\subsubsection{Verification Technique}\label{app:verification}
Our verification technique is identical to the one given by Gagie et al.~\cite{gagie2014lz77}
and involves a simple modification of the search for long primary occurrences.
By using Lemma~\ref{small-slp} instead of bookmarking~\cite{gagie2014lz77} for
extraction and fingerprinting and because we only need to verify $O(m/\tau)$ strings,
the verification procedure takes $O(m + (m/\tau) \lg(n/z))$ time
and uses $O(z\lg(n/z))$ space.

Consider the string $S$ of length $n$ that we wish to index and let $Z$ be the $LZ77$ parse of $S$.
The verification data structure is given by Lemma~\ref{small-slp}.
Consider the prefix search data structure $\tdp$ as given in Section~\ref{main-ds} and let $\phi$ be the fingerprinting function used
by the prefix search, the case for $\td$ is symmetric.
We alter the search for primary occurrences such that it first does the $O(m/\tau)$ prefix searches,
then verifies the results and discards false-positives before moving on to do the $O(m/\tau)$ range reporting queries
on the verified results. We also modify $\phi$ using Lemma~\ref{2power} to be collision-free for all
substrings of the indexed strings which length is a power of two.

Let $Q_1, Q_2, \ldots Q_j$
be all the suffixes of $P$ for which the prefix search found a locus candidate, let the candidates be
$v_1, v_2, \ldots v_j \in \tdp$ and let $p_i$ be  $\str{v_i}[1,|Q_i|]$.
Assume that $|Q_i| < |Q_{i+1}|$, and let \suf{Q} and \pref{Q} denote the fingerprints using $\phi$ of the suffix and prefix
respectively of length $2^{\floor{\lg |Q|}}$ of some string $Q$.
The verification progresses in iterations. Initially, let $a = 1$, $b = 2$
and for each iteration do as follows:
\begin{enumerate}
    \item $\suf{Q_a} \neq \suf{p_a}$ or  $\pref{Q_a} \neq \pref{p_a}$: Discard $v_a$ and set $a = a + 1$ and $b = b + 1$.
    \item $\suf{Q_a} = \suf{p_a}$ and $\pref{Q_a} = \pref{p_a}$, let $R = p_b[|p_b| - |p_a| + 1,|p_b|]$.
        \begin{enumerate}
            \item $\suf{R} = \suf{Q_a}$ and $\pref{R} = \pref{Q_a}$: set $a = a + 1$ and $b = b + 1$.
            \item $\suf{R} \neq \suf{Q_a}$ or $\pref{R} \neq \pref{Q_a}$: discard $v_b$ and set $b = b + 1$.
        \end{enumerate}

    \item $b = j + 1$: If all vertices have been discarded, report no matches.
        Otherwise, let $v_f$ be the last vertex in the sequence $v_1, \ldots, v_j$ that was not discarded.
        Report all non-discarded vertices $v_i$ where $|p_i|$ is no longer
        than the longest common suffix of $p_f$ and $Q_f$
        as verified and discard the rest.
\end{enumerate}

Consider the correctness and complexity of the algorithm.
In case 1, clearly, $p_a$ does not match $Q_a$ and thus $v_a$ must be a false-positive.
Now observe that because $Q_i$ is a suffix of $P$, it is also a suffix of $Q_{i'}$ for any $i < i'$.
Thus in case 2 (b), if $R$ does not match $Q_a$ then $v_b$ must be a false-positive.
In case 2 (a), both $v_a$ and $v_b$ may still be false-positives, yet
by Lemma~\ref{2power}, $p_a$ is a suffix of $p_b$ because $\suf{p_a} = \suf{R}$ and $\pref{p_a} = \pref{R}$.
Finally, in case $3$, $v_f$ is a true positive if and only if $p_f = Q_f$.
But any other non-discarded vertex $v_i \neq v_f$ is also only a true positive if $p_f$ and $Q_f$
share a length $|p_i|$ suffix because $p_i$ is a suffix of $p_f$ and $Q_i$ is a suffix of $Q_f$.

The algorithm does $j$ iterations and fingerprints of substrings of $P$
can be computed in constant time after $O(m)$ preprocessing.
Every vertex $v \in \tdp$ represents one or more substrings of $S$.
If we store the starting index in $S$ of one of these substrings in $v$ when constructing $\tdp$
we can compute the fingerprint of any substring $\str{v}[i,j]$ by computing the fingerprint of $S[i' + i - 1, i' + j - 1]$
where $i'$ is the starting index of one of the substrings of $S$ that $v$ represents.
By Lemma~\ref{small-slp}, the fingerprint computations take $O(\lg(n/z))$ time,
the longest common suffix of $p_f$ and $Q_f$ can be found in $O(m + \lg(n/z))$ time
and because $j \leq m/\tau$ the total time complexity of the algorithm is $O(m + (m/\tau) \lg(n/z))$.

\section{Short Primary Occurrences}
\label{short}

We now describe a simple data structure
that can find primary occurrences of $P$ in time $O(m + \occ)$ using space $O(z\tau)$
whenever $m \leq \tau$ where $\tau$ is a positive integer.

Let $Z$ be the LZ77 parse of the string $S$ of length $n$.
Let $Z[i] = S[s_i, e_i]$ and define $F$ to
be the union of the strings $S[k, \min\{e_i + \tau - 1, n\}]$ where $\max\{1, s_i, e_i -\tau + 1\} \leq k \leq e_i$ for $i = 1, 2, \ldots z$.
There are $O(z\tau)$ such strings, each of length $O(\tau)$
and they are all suffixes of the $z$ length $2\tau - 1$ substrings of $S$
starting $\tau - 1$ positions before each border position.
We store these substrings
along with the compact trie $\tf$ over the strings in $F$.
The edge labels of $\tf$ are compactly represented by
storing references into one of the substrings.
Every leaf stores the starting positions in $S$ of all the string it represents
and the positions of the leftmost borders these strings contain.

The combined size of $\tf$ and the substrings we store is $O(z\tau)$ and we simply
search for $P$ by navigating vertices using perfect hashing~\cite{Fredman} and matching
edge labels character by character.
Now either $\locus{P} = \bot$ in which case there are no primary occurrences of $P$ in $S$;
otherwise, $\locus{P} = v$ for some vertex $v \in \tf$ and thus every leaf in the subtree of $v$
represents a substring of $S$ that is prefixed by $P$.
By using the indices stored with the leaves, we can determine
the starting position for each occurrence and if it is primary or secondary.
Because each of the strings in $F$ start at different positions in $S$,
we will only find an occurrence once. Also, it is easy to see that
we will find all primary occurrences because of how the strings in $F$ are chosen.
It follows that the time complexity is $O(m + \occ)$ where $\occ$ is the number of primary and secondary occurrences.

\section{The Secondary Index}

Let $Z$ be the LZ77 parse of length $z$ representing the string $S$ of length $n$.
We find the secondary occurrences by
applying the most recent range reporting data structure by Chan~et~al.~\cite{Larsen11}
to the technique described by K\"{a}rkk\"{a}inen and Ukkonen~\cite{Karkkainen95lempel-zivparsing}
which is inspired by the ideas of Farach and Thorup~\cite{FarachT98}.

Let $o_1, \ldots o_{\occ}$ be the starting positions of the occurrences of $P$ in $S$ ordered increasingly.
Assume that $o_h$ is a secondary occurrence such that $P = S[o_h, o_h + m - 1]$. Then by definition,
$S[o_h, o_h + m - 1]$ is a substring the prefix $S[i, j-1]$ of some phrase $S[i,j]$
and there must be an occurrence of $P$ in the source of that phrase.
More precise, let $S[k, l] = S[i, j-1]$ be the source of the phrase $S[i,j]$
then $o_{h'} = k + o_h-i$ is an occurrence of $P$ for some $h' < h$.
We say that $o_{h'}$, which may be primary or secondary, is the source occurrence of the secondary occurrence $o_h$
given the LZ77 parse of $S$. Thus every secondary occurrence has a source occurrence.
Note that it follows from the definition that no primary occurrence has a source occurrence.

We find the secondary occurrences as follows:
Build a range reporting data structure $Q$ on the
$n \times n$ grid and if $S[i,j]$ is a phrase with source $S[i',j']$ we plot a point $(i',j')$
and along with it we store the phrase start $i$.

Now for each primary occurrence $o$ found by the primary index, we query $Q$ for the rectangle $[0,o] \times [o + m -1, n]$.
The points returned are exactly the occurrences having $o$ as source.
For each point $(x,y)$ and phrase start $i$ reported,
we report an occurrence $o' = i + o - x$ and recurse on $o'$ to find all the occurrences having $o'$ as source.

Because no primary occurrence have a source, while all secondary occurrences have a source,
we will find exactly the secondary occurrences.

The range reporting structure $Q$ is built using Lemma~\ref{lemma:larsen}
with $B = 2$ and uses space $O(z\lg\lg z)$.
Exactly one range reporting query is done for each primary and secondary occurrence
each taking $O((1 + k)\lg \lg n)$ where $k$ is the number of points reported.
Each reported point identifies a secondary occurrence, so the total time is $O(\occ \lg \lg n)$.

\section{The Compressed Index}
\label{combining}
We obtain our final index by combining the primary index, the verification data structure and the secondary index.
We use a standard technique to guarantee that no phrase in the LZ77 parse is longer than $n/z$ when building our primary index,
see e.g.~\cite{Rytter, Charikar}.
Therefore any primary occurrence of $P$ will have a prefix $P[1,k]$ where $k \leq n/z$
that is a suffix of some phrase. It then follows that
we need only consider the multiples $(P[1,i\tau],P[i\tau+1, m])$ for $i < \floor{\frac{n/z}{\tau}}$ when
searching for long primary occurrences.
This yields the following complexities:
\begin{itemize}
    \item $O(m + \frac{\min\{m, n/z\}}{\tau}(\frac{m}{x} + \lg x + B\lg\lg(z\tau)) + \occ\ B\lg\lg (z\tau))$ time
        and $O(z\tau\lg_B\lg(z\tau))$ space for the index finding long primary occurrences
        where $x$ and $\tau$ are positive integers and $2 \leq B \leq \lg^\epsilon(z\tau)$.

    \item $O(m + \occ)$ time and $O(z\tau)$ space for the index finding short primary occurrences.

    \item $O(m + (m/\tau) \lg(n/z))$ time and $O(z\lg(n/z))$ space for the verification data structure.

    \item $O(\occ \lg\lg n)$ time and $O(z\lg\lg z)$ space for the secondary index.
 \end{itemize}
 If we fix $x$ at $n/z$ we have $\frac{\min\{m, n/z\}}{\tau}\frac{m}{x} \leq m$ in which case we obtain
 the following trade-off simply by combining the above complexities.

\begin{theorem}
\label{trade-off}
    Given a string $S$ of length $n$ from an alphabet of size $\sigma$, we can build a compressed-index supporting substring queries
    in $O(m + \frac{m}{\tau}(\lg(n/z) + B\lg\lg(z\tau)) + \occ(B\lg\lg(z\tau) + \lg\lg n))$
    time using $O(z(\lg(n/z) + \tau\lg_B\lg(z\tau) + \lg\lg z))$ space
    for any query pattern $P$ of length $m$ where $2 \leq B \leq \lg^\epsilon(z\tau)$ and $0 < \epsilon < 1$
    are constants, $\tau$ is a positive integer, $z$ is the number of phrases in the LZ77 parse of $S$
    and $\occ$ is the number of occurrences of $P$ in $S$.
\end{theorem}

We note that none of our data structures assume constant sized alphabet
and thus Thm.~\ref{trade-off} holds for any alphabet size.

\subsection{Trade-offs}\label{app:tradeoff}

Thm.~\ref{trade-off} gives rise to a series of interesting time-space trade-offs.

\begin{corollary}
\label{theorem:second}
    Given a string $S$ of length $n$ from an alphabet of size $\sigma$ we can build a compressed-index supporting substring queries in
    \begin{enumerate}
        \item[(i)] $O(m(1 + \frac{\lg\lg z}{\lg(n/z)}) + \occ \lg\lg n)$ time using $O(z\lg(n/z)\lg\lg z)$ space, or
        \item[(ii)] $O(m(1 + \frac{\lg^\epsilon z}{\lg(n/z)}) + \occ(\lg\lg n + \lg^\epsilon z))$ time using $O(z(\lg(n/z) + \lg\lg z))$ space, or
        \item[(iii)] $O(m + \occ \lg\lg n)$ time using $O(z(\lg(n/z)\lg\lg z +  \lg^2\lg z))$ space, or
        \item[(iv)] $O(m + \occ(\lg\lg n + \lg^{\epsilon} z))$ time using $O(z(\lg(n/z) + \lg^{\epsilon} z))$ space.
    \end{enumerate}
    where $\epsilon > 0$ is an arbitrarily small constant.
\end{corollary}

\begin{proof}
    For $(i)$ set $B = 2$ and $\tau = \lg(n/z)$, for $(ii)$ set $B = \lg^{\epsilon} z$ and $\tau = \lg(n/z)$,
    for $(iii)$ set $B = 2$ and $\tau = \lg(n/z) + \lg \lg z$, for $(iv)$ set $B = \lg^{\epsilon'} z$ and $\tau = \lg(n/z) + \lg^{\epsilon} z$
    where $\epsilon' < \epsilon$.
\end{proof}

The leading term in the time complexity of Cor.~\ref{theorem:second} $(i)$ is $O(m)$ whenever $\lg\lg(z) = O(\lg(n/z))$
which is true when $z = O(n/\lg n)$, i.e.\ for all strings that are compressible by at least a logarithmic fraction.
For $\sigma = O(1)$ we have $z = O(n/\lg n)$ for all strings~\cite{Navarro} and thus Thm.~\ref{theorem:main} (i) follows immediately.
Cor.~\ref{theorem:second} $(ii)$ matches previous best space bounds
but obtains a leading term of $O(m)$ for any polynomial compression
rate. Thm.~\ref{theorem:main} $(ii)$ assumes constant-sized alphabet and therefore follows from (ii).
Cor.~\ref{theorem:second} $(iii)$ and $(iv)$ show how to guarantee the fast query times with leading term $O(m)$
without the assumptions on compression ratio that $(i)$ and $(ii)$ require to match this,
but at the cost of increased space.
Thm.~\ref{theorem:main} $(iii)$ is Cor.~\ref{theorem:second} $(ii)$ and 
thm.~\ref{theorem:main} $(iv)$ is Cor.~\ref{theorem:second} $(iv)$.

\subsection{Preprocessing}\label{app:preproc}

We now consider the preprocessing time of the data structure.
Let $Z$ be the LZ77 parse of the string $S$ of length $n$
let \td\ and \tdp\ be the compact tries used in the
index for long primary occurrences.
The compact trie \td\ indexes $O(z\tau)$ substrings of
$S$ with overall length $O(n\tau)$.
Thus we can construct the trie in $O(n\tau)$ time by
sorting the strings and successively inserting them in their sorted order~\cite{Andersson}.
The compact tries \tdp\ indexes $z\tau < n$ suffixes of $S$
and can be built in $O(n)$ time using $O(n)$ space~\cite{Farach}.
The index for short primary occurrences is a generalized suffix tree
over $z$ strings of length $O(\tau)$ with total length $z\tau < n$
and is therefore also built in $O(n)$ time.
The dictionaries used by the prefix search data structures
and for trie navigation contain $O(z\tau)$ keys
and are built in expected linear time using perfect hashing~\cite{Fredman}.
The range reporting data structures used by the primary and secondary index
over $O(z\tau)$ points are built in $O(z\tau \lg (z\tau))$ expected time using Lemma~\ref{lemma:larsen}.

Building the SLP for our verification data structure
takes $O(z\lg(n/z))$ time using Lemma~\ref{rytter}
and finding an appropriate fingerprinting function $\phi$
takes $O(n\lg n)$ expected time using Lemma~\ref{2power}.
The prefix search data structures \td\ and \tdp\ also require that
$\phi$ is collision-free for all prefixes whose length
are either pseudo fat or multiples of $x$.
There are at most $O(z\tau\lg n + n\tau/x)$ such prefixes~\cite{BelazzouguiBPV10}.
If we compute these fingerprints incrementally while doing a traversal of the tries,
we expect all the fingerprints to be unique.
We simply check this by sorting the fingerprints in linear time
and checking for duplicates by doing a linear scan.
If we choose a prime $p = \Theta(n^5)$
with Lemma~\ref{karp2}
then the probability of a collision between any two strings
is $O(1/n^4)$ and by a union bound
over the $O({(n \lg n)}^2)$ possible collisions
the probability that $\phi$ is collision-free is at least $1 - 1/n$.
Thus the expected time to find our required fingerprinting function is $O(n + n\lg n)$.

All in all, the preprocessing time for our combined index is therefore expected $O(n \lg n + n\tau)$.

\bibliography{bib}

\end{document}